\documentclass[conference]{IEEEtran}
\IEEEoverridecommandlockouts

\usepackage{cite}
\usepackage{amsmath,amssymb,amsfonts}
\usepackage{amsthm}
\usepackage{mathrsfs}
\usepackage{graphicx}
\usepackage{textcomp}
\usepackage{xcolor}
\usepackage{framed}
\usepackage{bbm}
\usepackage{manyfoot}
\usepackage{booktabs}
\usepackage{algorithm}
\usepackage{algorithmicx}
\usepackage{algpseudocode}
\usepackage{listings}
\usepackage{boxedminipage}
\usepackage{multirow}
\usepackage{dsfont}
\usepackage{makecell}
\usepackage{tabularx} % 表格美化+符号必备

\usepackage[top=0.8in, bottom=0.8in, left=0.75in, right=0.75in]{geometry}

\newtheorem{theorem}{Theorem}

\newtheorem{definition}{Definition}
\newtheorem{remark}{Remark}

\def\BibTeX{{\rm B\kern-.05em{\sc i\kern-.025em b}\kern-.08em
    T\kern-.1667em\lower.7ex\hbox{E}\kern-.125emX}}

\begin{document}

\title{Information-Theoretic Authenticated PIR: \\ From PIR-RV To APIR
\thanks{Pengzhen Ke, Yuxuan Qin, and Liang Feng Zhang are with the School of Information Science and Technology, ShanghaiTech University, Shanghai, China. Email: $\{$kepzh, qinyx2022, zhanglf$\}$@shanghaitech.edu.cn. These authors contributed equally: Pengzhen Ke, Yuxuan Qin.
This work was supported in part by the National Natural Science Foundation of China (No. 62372299).
% This work was supported in part by 
}
}
\author{\IEEEauthorblockN{Pengzhen Ke,~~ Yuxuan Qin,~~ Liang Feng Zhang}
}

\maketitle

\begin{abstract}
% THIS PAPER IS ELIGIBLE FOR THE STUDENT PAPER AWARD. 
Private Information Retrieval (PIR) allows clients to retrieve database entries without leaking retrieval indices, yet malicious servers seriously compromise retrieval correctness. Existing Authenticated PIR (APIR) schemes resist selective-failure attacks but rely on computational hardness assumptions. In contrast, information-theoretic PIR with Result Verification (itPIR-RV) achieves integrity without computational assumptions, yet only provides relaxed query privacy with no defense against selective-failure attacks.
This paper focuses on unconditionally secure information-theoretic APIR (itAPIR) constructions. We propose the rigorous information-theoretic security definition for itAPIR with statistical privacy against selective-failure attacks and integrity as core properties, formalize the hierarchical relation between itAPIR and itPIR-RV as a relaxed variant with identical integrity but basic query privacy, and prove a conversion theorem that valid itPIR-RV schemes can be directly upgraded to secure itAPIR with no extra overhead. 
Our work bridges the theoretical gap, simplifies itAPIR design, and enables quantum-resistant PIR in malicious server environments.
\end{abstract}

\section{Introduction}

Private Information Retrieval (PIR) \cite{chor1998private} allows a client to retrieve a target database block \( x_{\alpha} \in {\bf x} = (x_1, \ldots, x_n) \) while preserving the privacy of the retrieval index \( \alpha \in [n] \) from the involved servers. 
For large-scale databases, a naive PIR scheme incurring $O(n)$ communication complexity is impractical.
Thus, multi-server information-theoretic PIR has become a leading research direction \cite{beimel2001information, beimel2002robust, yang2002private, goldberg2007improving, yekhanin2008towards, efremenko20093, devet2012optimally, zhang2014verifiable, sun2017capacity, sun2017capacity2, banawan2018capacity, kurosawa2019correct, zhao2021verifiable, ben2022verifiable, eriguchi2022optimal, ke2022two, zhang2022byzantine, zhu2022post, colombo2023authenticated, ke2023private, kruglik2023two, eriguchi2024efficient, alon2025definition, li2025efficient, zhang2025unified, ke2026efficient }. In this model, the database is replicated across multiple servers, 
{
and the notion of \( t \)-privacy ensures that no coalition of up to \( t \) servers can infer the client's retrieval index \( \alpha \).
}

While early PIR protocols typically assume \emph{honest-but-curious} servers, a critical challenge in practical multi-server PIR deployments lies in addressing \emph{malicious} servers. Such servers may collude with each other to return incorrect responses (due to active attacks, network anomalies, or outdated data), aiming to deceive the client into reconstructing an erroneous value rather than the target $x_\alpha$. This concern has spurred extensive research on PIR protocols resilient to malicious servers 
\cite{beimel2002robust, yang2002private, goldberg2007improving, devet2012optimally, zhang2014verifiable, sun2017capacity2, banawan2018capacity, kurosawa2019correct, zhao2021verifiable, ben2022verifiable, ke2022two, zhang2022byzantine, zhu2022post, colombo2023authenticated, ke2023private, kruglik2023two, alon2025definition, ke2026efficient}, which are of particular interest in cloud computing scenarios.

Among existing solutions for malicious server tolerance, two representative types of verifiable PIR protocols stand out, both equipped with \emph{integrity} guarantees (i.e., they ensure the client can either recover the correct $x_\alpha$ or output $\perp$ to indicate invalid responses): Authenticated PIR (APIR) and PIR with Result Verification (PIR-RV). As a state-of-the-art error-detecting PIR scheme, APIR \cite{colombo2023authenticated} achieves subpolynomial communication complexity via Distributed Point Functions (DPFs) and features a strengthened privacy guarantee—resistance against \emph{selective-failure attacks}. Specifically, APIR’s privacy holds even if colluding malicious servers learn whether the client accepts the result (outputs $x_\alpha$) or rejects it (outputs $\perp$). In contrast, PIR-RV \cite{ke2022two, ke2023private}, a relaxed variant of APIR, does not require resistance to selective-failure attacks; its privacy guarantee only ensures that colluding servers cannot infer $\alpha$ from query distributions alone, without accounting for the client’s acceptance/rejection outcome.

Notably, existing APIR schemes are solely constructed under the \emph{computational security} model, relying on cryptographic hardness assumptions (e.g., DPF security under computational assumptions). This limits their applicability in high-security scenarios (e.g., secure multi-party computation, quantum-resistant systems) where computational assumptions may be compromised. To address this gap, this work focuses on constructing \emph{information-theoretic (it-) APIR} schemes—schemes that achieve unconditional security without relying on any computational assumptions. As a core contribution, we further establish a formal connection between itPIR-RV and itAPIR: we prove that an itPIR-RV scheme satisfying specific integrity conditions can be upgraded to an itAPIR scheme, thereby providing a constructive path to realizing information-theoretically secure APIR.

\subsection{Related Works}
\subsubsection{Information-Theoretic PIR}
Multi-server information-theoretic PIR (itPIR) was first formalized by Chor et al. \cite{chor1998private}, with subsequent works \cite{beimel2001information, yekhanin2008towards, efremenko20093} optimizing its communication complexity to sublinear or even polylogarithmic levels. These schemes achieve unconditional $t$-privacy but assume honest-but-curious servers, lacking resilience to malicious responses.

\subsubsection{Malicious-Resilient PIR}
To tolerate malicious servers, verifiable PIR schemes \cite{beimel2002robust, yang2002private, goldberg2007improving, devet2012optimally, zhang2014verifiable, sun2017capacity2, banawan2018capacity, kurosawa2019correct, zhao2021verifiable, ben2022verifiable, ke2022two, zhang2022byzantine, zhu2022post, eriguchi2022optimal, colombo2023authenticated, ke2023private,alon2025definition} introduce integrity guarantees, enabling the client to detect or even correct incorrect responses. 
Among them, itPIR-RV \cite{ke2022two, eriguchi2022optimal, ke2023private} achieves information-theoretic $(v,\epsilon)$-integrity { (i.e., the probability that the client accepts an incorrect response is at most $\epsilon$, enven if up to $v$ servers collude maliciously)} and only guarantees basic perfect or statistical privacy; APIR \cite{colombo2023authenticated} provides strengthened privacy with resistance to selective-failure attacks but relies on computational assumptions, for which no information-theoretic variant has been formally defined.

\subsubsection{APIR vs. PIR-RV}
Existing works treat APIR and PIR-RV as independent primitives, but their theoretical connection (especially in the information-theoretic setting) remains unstudied—this gap limits the construction of high-security itAPIR schemes.

{
% \color{blue}
\subsection{Our Contributions}
This work makes four key contributions to information-theoretic malicious-resilient PIR:
\begin{itemize}
    \item We propose information-theoretic Authenticated PIR (itAPIR), requiring statistical $t$-privacy against selective-failure attacks and $(v,\epsilon)$-integrity against malicious servers.

    \item We clarify the relation between itAPIR and itPIR-RV, showing that itPIR-RV is a relaxed variant of itAPIR with the same integrity guarantee but without selective-failure-attack resilience.

    \item We prove a conversion theorem: any perfect/statistical $t$-private itPIR-RV scheme with $(v,\epsilon)$-integrity, where $v \ge t$ and $\epsilon$ is negligible, can be transformed into a statistically $t$-private itAPIR scheme.

    \item We instantiate the theorem with existing efficient itPIR-RV schemes to obtain two concrete itAPIR constructions for small-server and general-server settings, both with sublinear communication complexity.
\end{itemize}
To highlight the differences between prior schemes and our results, Table~\ref{tab:apir-comparison} summarizes representative malicious-resilient PIR constructions.
}

\begin{table}[t]
\centering
\caption{Comparison of representative malicious-resilient PIR schemes and our constructions. 
All schemes assume $\ell$ servers and tolerate $t$ corrupted servers for privacy.}
\label{tab:apir-comparison}
\renewcommand{\arraystretch}{1.2}

\begin{tabularx}{\columnwidth}{@{}*{4}{>{\centering\arraybackslash}X}@{}}

\toprule
{Scheme}
& \raisebox{0.3ex}{\thead{Corrupted Servers\\for Integrity ($v$)}} 
& \raisebox{0.3ex}{\thead{Selective-Failure\\Attack Resilience}} 
& {Security Model}\\
\midrule

\cite{ke2023private} 
& $\ell-1$ 
& $\times$ 
& Statistical 
\\

\cite{eriguchi2022optimal} 
& $t$ 
& $\times$ 
& Statistical 
\\

\cite{colombo2023authenticated} 
& $t$ 
& $\checkmark$ 
& Computational 
\\

\cite{alon2025definition} 
& $t$ 
& $\checkmark$ 
& Statistical 
\\

Theorem \ref{theo:itapir-construction1} 
& $\ell-1$ 
& $\checkmark$ 
& Statistical 
\\

Theorem \ref{theo:itapir-construction2} 
& $t$ 
& $\checkmark$ 
& Statistical 
\\

\bottomrule
\end{tabularx}
\end{table}

\section{Preliminaries} \label{preliminaries}
{\bf Notations.  }
$\mathbb{N}$ denotes the set of natural numbers $\{1,2,3,\ldots\}$.
For any positive integer $n$, we denote $[n] = \{1,2,\ldots, n\}$ and $\{a_j\}_{j\in [n]}=\{a_1, a_2, \ldots, a_n\}$.
The indicator function $\mathds{1}\{\mathcal{P}\}$ for a proposition $\mathcal{P}$, which outputs $1$ if $\mathcal{P}$ holds true and $0$ otherwise.
For two probability distributions $X$ and $Y$ over the same finite set:
$X \equiv Y$ means $X$ and $Y$ are \emph{perfectly equivalent} (identical probability distributions), i.e., their statistical distance $\Delta(X,Y) = 0$;
$X \approx_s Y$ means $X$ and $Y$ are \emph{statistically indistinguishable}, i.e., their statistical distance $\Delta(X,Y)$ is negligible in the security parameter.
    \vspace{2mm}

\subsection{Information-Theoretic Authenticated PIR} \label{sec:it-apir}
%%%%%%%%%%%%%%%%%%%%%%%%%%%%%%%%%%%%%%%%%%%%%%%%%%%%%%%%%%%
The concept of Authenticated PIR \cite{colombo2023authenticated} was originally introduced with a computational security level. In this work, we present its definition in the \textit{information-theoretic setting}.

Informally, an $\ell$-server information-theoretic APIR (itAPIR) scheme involves $\ell$ servers $\{\mathcal{S}_j\}_{j \in [\ell]}$, each storing a copy of the same database ${\bf x} = (x_1, \ldots, x_n) \in (\{0,1\}^m)^n$ (each entry $x_i$ is an $m$-bit block), and a client seeking to retrieve $x_{\alpha}$ for some $\alpha \in [n]$. 
The scheme guarantees that the client can correctly recover $x_{\alpha}$ when all $\ell$ servers $\{\mathcal{S}_j\}_{j \in [\ell]}$ are honest and respond correctly, or output a special symbol $\perp$ to indicate incorrect responses. The scheme is said to be $t$-private in the information-theoretic sense: any coalition of up to $t$ malicious servers gains statistically no information about the client's retrieval index $\alpha$, even if the adversary learns whether the client accepts the retrieved result or outputs the rejection symbol $\perp$.

\begin{definition}[itAPIR] \label{def:itapir-scheme}
An $\ell$-server itAPIR scheme $\Gamma = ({\sf Que}, {\sf Ans}, {\sf Rec})$ consists of three algorithms executed by the client and servers as specified below:
\begin{itemize}    
    \item $(\{q_j\}_{j \in [\ell]}, {\sf aux}) \leftarrow {\sf Que}(1^{\kappa}, n, \alpha)$:  
    A randomized \emph{querying algorithm} executed by the client. It takes the security parameter $\kappa$, the database size $n$, and retrieval index $\alpha \in [n]$ as input, and outputs $\ell$ queries $\{q_j\}_{j \in [\ell]}$ (with $q_j$ sent to server $\mathcal{S}_j$) and auxiliary information ${\sf aux}$ for reconstruction.

    \item $a_j \leftarrow {\sf Ans}({\bf x}, q_j)$:  
    A deterministic \emph{answering algorithm} executed by server $\mathcal{S}_j$ ($j \in [\ell]$). It takes the database ${\bf x} = (x_1, \ldots, x_n) \in (\{0,1\}^m)^n$ and query $q_j$ as input, and outputs a response $a_j$. 
    
    \item $y \leftarrow {\sf Rec}(\{a_j\}_{j \in [\ell]}, {\sf aux})$:  
    A deterministic \emph{reconstructing algorithm} executed by the client. It takes the responses $\{a_j\}_{j \in [\ell]}$ and auxiliary information ${\sf aux}$ as input, and outputs either the correct value $y = x_{\alpha}$ (when all responses are valid) or the symbol $y = \perp$ (indicating at least one invalid response).
\end{itemize}
\end{definition}

To formalize the above goals, an itAPIR scheme $\Gamma$ must satisfy the following three properties:

\begin{definition}[itAPIR Correctness] \label{def:itapir-correctness}
Informally, $\Gamma$ is \emph{correct} if the reconstruction algorithm $\sf Rec$ outputs the correct value $x_\alpha$ when all $\ell$ servers respond honestly.  

Formally, for any security parameter $\kappa \in \mathbb{N}$, database size $n \in \mathbb{N}$, data entry size $m \in \mathbb{N}$, database ${\bf x} \in (\{0,1\}^m)^n$, any retrieval index $\alpha \in [n]$, and query-auxiliary pair $(\{q_j\}_{j \in [\ell]}, {\sf aux}) \leftarrow {\sf Que}(1^\kappa, n, \alpha)$, let $y \leftarrow {\sf Rec}\bigl(\{ {\sf Ans}({\bf x}, q_j)\}_{j \in [\ell]}, {\sf aux}\bigr)$. Then:
\[
\Pr\left[ y = x_\alpha \right] = 1.
\]
\end{definition}

\begin{definition}[itAPIR Statistical Privacy] \label{def:itapir-privacy}
Informally, $\Gamma$ is \emph{statistically $t$-private} if any coalition of up to $t$ malicious servers gains no information about the client's retrieval index $\alpha$, \emph{even if the adversary learns whether the client accepts the result or rejects it (outputs $\perp$)}. This captures security against selective-failure attacks, a critical strengthened privacy property of APIR.

Formally, let the subset ${\bf T} \subseteq [\ell]$ with cardinality $|{\bf T}| \leq t$ be the set of the indices of corrupted servers, that is, the adversary controls the servers $\{{\cal S}_{j}\}_{j\in {\bf T}}$. For any probabilistic computationally unbounded { two-stage} adversary $\mathcal{A}$ { with two polynomial-time algorithms $\mathcal{A}_0, \mathcal{A}_1$}, any database size $n\in \mathbb{N}$, any data entry size $m \in \mathbb{N}$, any database ${\bf x} \in (\{0,1\}^m)^n$, and any retrieval index $\alpha \in [n]$, define the real experiment:
\[
\mathsf{REAL}_{\alpha} = \left\{
\beta : 
\begin{aligned}
&(\{q_j\}_{j \in [\ell]}, {\sf aux}) \leftarrow {\sf Que}(1^\kappa, n, \alpha), \\
&(\mathsf{st}_{\mathcal{A}}, \{\tilde{a}_j\}_{j \in {\bf T}}) \leftarrow \mathcal{A}_0(\mathbf{x}, \{q_j\}_{j \in {\bf T}} ), \\
&a_j \leftarrow {\sf Ans}({\bf x}, q_j) \quad \forall j \in [\ell] \setminus {\bf T}, \\
&y \leftarrow {\sf Rec}(\{a_j\}_{j \notin {\bf T}} \cup \{\tilde{a}_j\}_{j \in {\bf T}}, {\sf aux}), \\
&b \leftarrow \mathds{1}\{ y \neq \perp \}, \\
&\beta \leftarrow \mathcal{A}_1(\mathsf{st}_{\mathcal{A}}, b)
\end{aligned}
\right\}.
\]
The scheme is statistically $t$-private if there exists a { two-stage} simulator $\mathsf{Sim}$ { with two polynomial algorithm} $\mathsf{Sim}_0, \mathsf{Sim}_1$ such that $\mathsf{REAL}_{\alpha} \approx_{s} \mathsf{IDEAL}_{\mathsf{Sim}}$ for all $\alpha \in [n]$, where $\mathsf{IDEAL}_{\mathsf{Sim}}$ is defined as:
\[
\mathsf{IDEAL}_{\mathsf{Sim}} = \left\{
\beta : 
\begin{aligned}
&(\mathsf{st}_{\mathsf{Sim}}, \{q_j'\}_{j \in {\bf T}}) \\
&\qquad \leftarrow \mathsf{Sim}_0(1^\kappa, n, \mathbf{x}, {\bf T}), \\
&(\mathsf{st}_{\mathcal{A}}, \{ \tilde{a}_j \}_{j \in {\bf T}}) \\
&\qquad \leftarrow \mathcal{A}_0(\mathbf{x}, \{q_j'\}_{j \in {\bf T}}), \\
&b \leftarrow \mathsf{Sim}_1(\mathsf{st}_{\mathsf{Sim}}, \{ \tilde{a}_j \}_{j \in {\bf T}}), \\
&\beta \leftarrow \mathcal{A}_1(\mathsf{st}_{\mathcal{A}}, b)
\end{aligned}
\right\}.
\]
\end{definition}
Regarding the absence of a ``perfectly private itAPIR‘’ definition, we note that such a notion (requiring $\mathsf{REAL}_\alpha \equiv \mathsf{IDEAL}_{\mathsf{Sim}}$ for all $\alpha$) is logically feasible in theory. However, our core conversion theorem (Theorem \ref{theo:pir-rv-to-apir}) reveals a fundamental constraint: the $(v, \epsilon)$-integrity of itPIR-RV introduces a negligible statistical error $\epsilon$, which propagates to the privacy guarantee of the upgraded itAPIR scheme. This makes perfect privacy unattainable via the itPIR-RV upgrade path. Thus, we focus only on statistical privacy for itAPIR in this work.

\begin{definition}[itAPIR Integrity] \label{def:itapir-integrity}
Informally, the scheme $\Gamma$ satisfies $(v, \epsilon)$-\emph{integrity} if no coalition of up to $v$ servers can deceive the client with retrieval index $\alpha$ into outputting a result $y \notin \{x_{\alpha}, \perp\}$ by providing fraudulent responses.  

Formally, for any adversary $\mathcal{A}$ controlling a set of servers ${\bf V} \subseteq [\ell]$ of cardinality at most $v$, any database ${\bf x}\in (\{0,1\}^m)^n$, and any index $\alpha \in [n]$, the probability that the client outputs an incorrect value is bounded by $\epsilon$:
\[
\Pr\left[ 
\begin{aligned}
y \notin \{x_\alpha,& \perp\}:\\ 
&(\{q_j\}_{j \in [\ell]}, {\sf aux}) \leftarrow {\sf Que}(1^\kappa, n, \alpha), \\
&\{ \tilde{a}_j \}_{j \in {\bf V}} \leftarrow \mathcal{A}({\bf x}, \{q_j\}_{j \in {\bf V}}), \\
&a_j \leftarrow {\sf Ans}({\bf x}, q_j) \quad \forall j \notin {\bf V}, \\
&y \leftarrow {\sf Rec}(\{a_j\}_{j \notin {\bf V}} \cup \{ \tilde{a}_j \}_{j \in {\bf V}}, {\sf aux})
\end{aligned}
\right] \leq \epsilon.
\]
\end{definition}

\subsection{PIR with Result Verification} \label{sec:pir-rv}
%%%%%%%%%%%%%%%%%%%%%%%%%%%%%%%%%%%%%%%%%%%%%%%%%%%%%%%%%%%
An information-theoretic $\ell$-server PIR with Result Verification (itPIR-RV) \cite{ke2022two, eriguchi2022optimal, ke2023private, li2025efficient, eriguchi2024efficient} scheme is a \textbf{relaxed variant} of the information-theoretic APIR (itAPIR) scheme (Section \ref{sec:it-apir}).
It shares the same correctness and $(v, \epsilon)$-integrity properties as itAPIR. The only key difference lies in privacy: unlike itAPIR (which requires statistical privacy against selective-failure attacks), itPIR-RV does \textbf{not} guarantee resistance to selective-failure attacks. 
Its privacy only ensures that colluding servers gain no information about $\alpha$ from query distributions alone (ignoring the client’s acceptance/rejection outcome), and includes two information-theoretic flavors: perfect privacy and statistical privacy.

\begin{definition}[itPIR-RV Correctness] \label{def:itpirrv-correctness}
Informally, $\Pi$ is \emph{correct} if the reconstruction algorithm $\sf Rec$ outputs the correct value $x_\alpha$ when all $\ell$ servers respond honestly.  

This property is identical to itAPIR correctness (Definition \ref{def:itapir-correctness}).
\end{definition}

\begin{definition}[itPIR-RV Perfect Privacy] \label{def:itpirrv-privacy-perfect}
Informally, an itPIR-RV scheme $\Pi = ({\sf Que}, {\sf Ans}, {\sf Rec})$ is \emph{perfectly $t$-private} if any coalition of up to $t$ servers gains \textbf{absolutely no information} about the client's retrieval index $\alpha$ from query distributions alone. No adversary—even with unbounded computational power—can distinguish query distributions for different $\alpha$.

Formally, for any security parameter $\kappa\in \mathbb{N}$, database size $n\in \mathbb{N}$, any index $\alpha \in [n]$, and any subset ${\bf T} \subseteq [\ell]$ with $|{\bf T}| \leq t$, let the query distribution:
\[
\mathsf{REAL}_{\alpha} = \left\{ \{q_j\}_{j\in {\bf T}} : (\{q_j\}_{j\in [\ell]}, {\sf aux})\leftarrow {\sf Que}(1^{\kappa}, n, \alpha)    \right\}.
\]
There exists a simulator $\mathsf{Sim}'$ such that for any $\alpha \in [n]$, the ideal distribution:
\[
\mathsf{IDEAL}_{\mathsf{Sim}'} = \left\{  \{q_j\}_{j \in {\bf T}} \leftarrow \mathsf{Sim}' (1^{\kappa}, n, {\bf T}) \right\},
\]
satisfies 
\[
    \mathsf{REAL}_{\alpha} \equiv \mathsf{IDEAL}_{\mathsf{Sim}'}.
\]
\end{definition}

\begin{definition}[itPIR-RV Statistical Privacy] \label{def:itpirrv-privacy-statistical}
Informally, an itPIR-RV scheme $\Pi = ({\sf Que}, {\sf Ans}, {\sf Rec})$ is \emph{statistically $t$-private} if any coalition of up to $t$ servers gains \textbf{statistically no information} about the client's retrieval index $\alpha$ from query distributions alone.

Formally, for any security parameter $\kappa\in \mathbb{N}$, database size $n\in \mathbb{N}$, any index $\alpha \in [n]$, and any subset ${\bf T} \subseteq [\ell]$ with $|{\bf T}| \leq t$, let $\mathsf{REAL}_{\alpha}$ and $\mathsf{IDEAL}_{\mathsf{Sim}'}$ be defined as in Definition \ref{def:itpirrv-privacy-perfect}. There exists a simulator $\mathsf{Sim}'$ such that for any $\alpha \in [n]$:
\[
    \mathsf{REAL}_{\alpha} \approx_{s} \mathsf{IDEAL}_{\mathsf{Sim}'}.
\]
\end{definition}

\begin{remark} \label{rem:itpirrv-privacy-relation}
Perfect $t$-privacy (Definition \ref{def:itpirrv-privacy-perfect}) is strictly stronger than statistical $t$-privacy (Definition \ref{def:itpirrv-privacy-statistical}): perfect distributional equivalence ($X \equiv Y$) implies statistical indistinguishability ($X \approx_s Y$), so any perfectly $t$-private itPIR-RV scheme is inherently statistically $t$-private.
An itPIR-RV scheme is deemed privacy-preserving if it satisfies either of these two definitions (only one is required).
\end{remark}

\begin{remark} \label{remark1}
The privacy definitions above (Definitions \ref{def:itpirrv-privacy-perfect} and \ref{def:itpirrv-privacy-statistical}) are equivalent to the standard formulation for information-theoretic PIR. Let $\sim$ denote $\equiv$ (perfect) or $\approx_s$ (statistical), where $\mathsf{REAL}_{\alpha_0} \sim \mathsf{REAL}_{\alpha_1}$ (for any distinct $\alpha_0, \alpha_1 \in [n]$) is the standard privacy condition.
Mutual equivalence holds:

(i) $\boldsymbol{(\Rightarrow)}$ Suppose a scheme satisfies either perfect or statistical $t$-privacy (per Definitions \ref{def:itpirrv-privacy-perfect}/\ref{def:itpirrv-privacy-statistical}), i.e., there exists a simulator $\mathsf{Sim}'$ such that $\mathsf{REAL}_{\alpha} \sim \mathsf{IDEAL}_{\mathsf{Sim}'}$ for all $\alpha \in [n]$. By transitivity of $\sim$ (equivalence or indistinguishability), we have $\mathsf{REAL}_{\alpha_0} \sim \mathsf{REAL}_{\alpha_1}$ for any distinct $\alpha_0, \alpha_1 \in [n]$.

(ii) $\boldsymbol{(\Leftarrow)}$ Suppose $\mathsf{REAL}_{\alpha_0} \sim \mathsf{REAL}_{\alpha_1}$ for any distinct $\alpha_0, \alpha_1 \in [n]$. Construct a simulator $\mathsf{Sim}'$ that samples from $\mathsf{REAL}_{\alpha}$ (for any fixed $\alpha \in [n]$). For any $\alpha' \in [n]$, $\mathsf{REAL}_{\alpha'} \sim \mathsf{REAL}_{\alpha}$, so $\mathsf{REAL}_{\alpha'} \sim \mathsf{IDEAL}_{\mathsf{Sim}'}$, satisfying the respective privacy definition.
\end{remark}

\begin{definition}[itPIR-RV Integrity] \label{def:itpirrv-integrity}
Informally, the scheme $\Pi$ satisfies $(v, \epsilon)$-\emph{integrity} if no coalition of up to $v$ servers can deceive the client with retrieval index $\alpha$ into outputting a result $y \notin \{x_{\alpha}, \perp\}$ by providing fraudulent responses.  

This property is identical to itAPIR integrity (Definition \ref{def:itapir-integrity}).
\end{definition}

\section{From itPIR-RV to itAPIR}
\label{sec:pir-rv-to-itapir}
%%%%%%%%%%%%%%%%%%%%%%%%%%%%%%%%%%%%%%%%%%%%%%%%%%%%%%%%%%%
In Sections \ref{sec:it-apir} and \ref{sec:pir-rv}, we defined itAPIR (Definition \ref{def:itapir-scheme}) and itPIR-RV (a relaxed itAPIR variant lacking selective-failure attack resistance). A natural question arises: \emph{Under what conditions can an itPIR-RV scheme upgrade to an itAPIR scheme, and how to derive practical itAPIR constructions}? This section addresses both via two contributions:

First, we prove a key conversion theorem: an itPIR-RV scheme satisfying specific integrity and privacy properties can be directly regarded as an itAPIR scheme, as it inherently achieves the strengthened privacy (resistance to selective-failure attacks) required by itAPIR. 
Second, using this theorem with existing efficient itPIR-RV schemes, we present two concrete itAPIR constructions to verify the theorem’s practicality.

Informally, the core insight of the conversion is as follows: The only gap between itPIR-RV and itAPIR lies in privacy—while itPIR-RV only guarantees query privacy (either perfect or statistical, itAPIR requires statistical privacy against selective-failure attacks. The $(v, \epsilon)$-integrity of itPIR-RV (with $v \geq t$) bridges this gap: it limits the probability that adversaries (controlling up to $t \leq v$ servers) can submit malicious answers that deceive the client into accepting an incorrect result (i.e., $y \notin \{x_\alpha, \perp\}$) to a negligible value $\epsilon$. 
Combining this integrity guarantee with the privacy of itPIR-RV, we show that the adversary’s view (including queries and the acceptance/rejection signal $b$) in the real itAPIR experiment is statistically indistinguishable from that in the ideal experiment—thus satisfying the full privacy requirement of itAPIR.

Formally, we state the conversion theorem below, followed by a rigorous proof to validate the above intuition.

\begin{theorem} \label{theo:pir-rv-to-apir} 
    Let $\Pi = ({\sf Que}, {\sf Ans}, {\sf Rec})$ be an $\ell$-server itPIR-RV scheme satisfying either perfect $t$-privacy (Definition \ref{def:itpirrv-privacy-perfect}) or statistical $t$-privacy (Definition \ref{def:itpirrv-privacy-statistical}). If $\Pi$ satisfies $(v, \epsilon)$-integrity (Definition \ref{def:itpirrv-integrity}) for $v \geq t$ and a negligible function $\epsilon$ in the security parameter $\kappa$, then $\Pi$ is an $\ell$-server statistically $t$-private itAPIR scheme that satisfies the selective-failure attack resistance required by itAPIR’s statistical privacy (Definition \ref{def:itapir-privacy}).
\end{theorem}

\begin{proof}
To prove the theorem, we only need to show that $\Pi$ satisfies the statistical privacy requirement of itAPIR, i.e., privacy against selective-failure attacks. 

For any probabilistic computationally unbounded adversary $\mathcal{A} = (\mathcal{A}_0, \mathcal{A}_1)$ (controlling a subset ${\bf T} \subseteq [\ell]$ of servers with $|{\bf T}| \leq t$), let a simulator $\mathsf{Sim} = (\mathsf{Sim}_0, \mathsf{Sim}_1)$ proceed as follows—using the underlying simulator $\mathsf{Sim}'$ of scheme $\Pi$ (from itPIR-RV’s privacy definition (either Definition \ref{def:itpirrv-privacy-perfect} or \ref{def:itpirrv-privacy-statistical})):

\hspace{1cm}

\noindent
\begin{tabular}{@{\hspace{1mm}}p{0.49\columnwidth}@{\hspace{1mm}}@{\hspace{5mm}}p{0.42\columnwidth}@{}}
$\mathsf{Sim}_0(1^\kappa, n, {\bf x}, {\bf T})$ & $\mathsf{Sim}_1(\mathsf{st}_{\mathsf{Sim}}, \{\tilde{a}_j\}_{j \in {\bf T}})$ \\
\hline
\vspace{-0.5em} 
\begin{tabular}{@{}l@{\hspace{0.5em}}p{0.5\columnwidth}@{}}
    1: & $\{q_j\}_{j \in {\bf T}} \leftarrow \mathsf{Sim}'(1^\kappa, n,  {\bf T})$, \\
    2: & $\forall j \in {\bf T},~ a_j \leftarrow {\sf Ans}({\bf x}, q_j)$, \\
    3: & $\mathsf{st}_{\mathsf{Sim}} \leftarrow \{a_j\}_{j \in {\bf T}}$, \\
    4: & \textbf{return} $(\mathsf{st}_{\mathsf{Sim}}, \{q_j\}_{j \in {\bf T}}).$
\end{tabular}
&
\vspace{-0.5em}
\begin{tabular}{@{}l@{\hspace{0.5em}}p{0.5\columnwidth}@{}}
    1: & $\{a_j\}_{j \in {\bf T}}\leftarrow \mathsf{st}_{\mathsf{Sim}}$, \\
    2: & $\forall j \in {\bf T},~ \delta_j = \tilde{a}_j - a_j$, \\
    3: & $b \leftarrow \mathds{1}\{ \wedge_{j\in {\bf T}} (\delta_j = 0) \}$, \\
    4: & \textbf{return} $b$.
\end{tabular}
\end{tabular}

\hspace{1cm}

To bridge the real and ideal privacy experiments and verify indistinguishability, we define a sequence of hybrid experiments $\mathrm{H}_0, \mathrm{H}_1, \mathrm{H}_2, \mathrm{H}_3$ as follows:

\begin{itemize}
    \item \textbf{$\mathrm{H}_0$}: This is the real privacy experiment $\mathsf{REAL}_{\alpha}$ (Definition \ref{def:itapir-privacy}). The challenger runs the protocol honestly for the retrieval index $\alpha$. $\mathcal{A}_0$ receives queries for the corrupted servers and returns potentially malformed answers. $\mathcal{A}_1$ receives a bit $b$ indicating whether the client accepted the result.
    \begin{align*}
    \mathrm{H}_0 = \left\{
    \beta : 
    \begin{aligned}
        &(\{q_j\}_{j \in [\ell]}, {\sf aux}) \leftarrow {\sf Que}(1^\kappa, n, \alpha), \\
        &\forall j \in [\ell] , ~~a_j \leftarrow {\sf Ans}({\bf x}, q_j), \\
        &(\mathsf{st}_{\mathcal{A}}, \{\tilde{a}_j\}_{j \in {\bf T}}) \leftarrow \mathcal{A}_0({\bf x}, \{q_j\}_{j \in {\bf T}}), \\
        &y \leftarrow {\sf Rec}(\{a_j\}_{j \in [\ell] \setminus {\bf T}} \cup \{\tilde{a}_j\}_{j \in {\bf T}}, {\sf aux}), \\
        &b \leftarrow \mathds{1}\{ y \neq \perp \}, \\
        &\beta \leftarrow \mathcal{A}_1(\mathsf{st}_{\mathcal{A}}, b)
    \end{aligned}
    \right\}.
    \end{align*}

    \item \textbf{$\mathrm{H}_1$}: Same as $\mathrm{H}_0$ except for the computation of the acceptance bit $b$: instead of determining $b$ via the client’s reconstruction algorithm ${\sf Rec}$ (checking if $y \neq \perp$), the challenger computes $b$ by verifying whether the adversary’s answers $\{\tilde{a}_j\}_{j \in {\bf T}}$ match the honest answers $\{a_j\}_{j \in {\bf T}}$ (generated from the same queries). The bit $b$ is set to $1$ if and only if all answers are identical.
    \begin{align*}
    \mathrm{H}_1 = \left\{
    \beta : 
    \begin{aligned}
        &(\{q_j\}_{j \in [\ell]}, {\sf aux}) \leftarrow {\sf Que}(1^\kappa, n, \alpha), \\
        &\forall j \in {\bf T} , ~~a_j \leftarrow {\sf Ans}({\bf x}, q_j), \\
        &(\mathsf{st}_{\mathcal{A}}, \{\tilde{a}_j\}_{j \in {\bf T}}) \leftarrow \mathcal{A}_0({\bf x}, \{q_j\}_{j \in {\bf T}}), \\
        &\forall j \in {\bf T}, ~~\delta_j = \tilde{a}_j - a_j, \\
        &b \leftarrow \mathds{1}\{ \wedge_{j\in {\bf T}} (\delta_j = 0) \}, \\
        &\beta \leftarrow \mathcal{A}_1(\mathsf{st}_{\mathcal{A}}, b)
    \end{aligned}
    \right\}.
    \end{align*}

    \item \textbf{$\mathrm{H}_2$}: Same as $\mathrm{H}_1$ except for the generation of queries $\{q_j\}_{j \in {\bf T}}$: instead of using the honest query algorithm  
    \begin{align*}
    \mathrm{H}_2 = \left\{
    \beta : 
    \begin{aligned}
        & \{q_j\}_{j \in {\bf T}} \leftarrow \mathsf{Sim'}(1^\kappa , n, {\bf T}), \\
        &\forall j \in {\bf T} , ~~a_j \leftarrow {\sf Ans}({\bf x}, q_j), \\
        &(\mathsf{st}_{\mathcal{A}}, \{\tilde{a}_j\}_{j \in {\bf T}}) \leftarrow \mathcal{A}_0({\bf x}, \{q_j\}_{j \in {\bf T}}), \\
        &\forall j \in {\bf T}, ~~\delta_j = \tilde{a}_j - a_j, \\
        &b \leftarrow \mathds{1} \left\{ \wedge_{j\in {\bf T}} (\delta_j = 0) \right\}, \\
        &\beta \leftarrow \mathcal{A}_1(\mathsf{st}_{\mathcal{A}}, b)
    \end{aligned}
    \right\}.
    \end{align*}    
    ${\sf Que}$ (which depends on the retrieval index $\alpha$), queries
    are generated via the simulator $\mathsf{Sim}'$ (from itPIR-RV’s privacy definition (either Definition \ref{def:itpirrv-privacy-perfect} or \ref{def:itpirrv-privacy-statistical})).

    \item \textbf{$\mathrm{H}_3$}: Identical in structure to $\mathrm{H}_2$, and corresponds to the ideal privacy experiment $\mathsf{IDEAL}_{\mathsf{Sim}}$ (Definition \ref{def:itapir-privacy}). This experiment operates without access to the client’s retrieval index $\alpha$.
    \begin{align*}
    \mathrm{H}_3 = \left\{
    \begin{aligned}
        \beta &: \\
        &(\mathsf{st}_{\mathsf{Sim}}, \{q_j\}_{j \in {\bf T}}) \leftarrow \mathsf{Sim}_0(1^\kappa, n, {\bf x}, {\bf T}), \\
        &(\mathsf{st}_{\mathcal{A}}, \{\tilde{a}_j\}_{j \in {\bf T}}) \leftarrow \mathcal{A}_0({\bf x}, \{q_j\}_{j \in {\bf T}}), \\
        &b \leftarrow \mathsf{Sim}_1(\mathsf{st}_{\mathsf{Sim}}, \{\tilde{a}_j\}_{j \in {\bf T}}), \\
        &\beta \leftarrow \mathcal{A}_1(\mathsf{st}_{\mathcal{A}}, b)
    \end{aligned}
    \right\}.
    \end{align*}
\end{itemize}

We now prove the indistinguishability of the hybrid experiments. To quantify this, we first define an event for each hybrid: For \( j \in \{0,1,2,3\} \), let \( W_j \) be the event that the output of the hybrid experiment \( \mathrm{H}_j \) is “1”. 

\begin{itemize}
    \item 
    Hybrid \( \mathrm{H}_1 \) is identical to \( \mathrm{H}_0 \) except for how the acceptance bit \( b \) is set: \( \mathrm{H}_0 \) sets \( b = \mathds{1}\{y \neq \perp\} \), while \( \mathrm{H}_1 \) sets \( b = \mathds{1}\{\wedge_{j\in {\bf T}} (\delta_j = 0)\} \).
    The two hybrids differ only if the adversary submits incorrect answers (\( \exists j \in {\bf T}, \tilde{a}_j \neq a_j \)) but the client’s \( {\sf Rec} \) algorithm fails to output \( \perp \) (i.e., the client accepts an incorrect value). By the $(v, \epsilon)$-integrity property of itPIR-RV (Definition \ref{def:itpirrv-integrity}), the probability of this event is bounded by \( \epsilon \). Thus:
    \[ |\Pr[W_0] - \Pr[W_1]| \leq \epsilon. \]

    \item 
    Hybrid \( \mathrm{H}_2 \) differs from \( \mathrm{H}_1 \) only in the generation of queries \( \{q_j\}_{j \in {\bf T}} \): \( \mathrm{H}_1 \) uses the real algorithm \( {\sf Que}(1^\kappa, n, \alpha) \), while \( \mathrm{H}_2 \) uses the simulator \( \mathsf{Sim}' \) from itPIR-RV’s privacy definition (either Definition \ref{def:itpirrv-privacy-perfect} or \ref{def:itpirrv-privacy-statistical}).

    By the itPIR-RV scheme’s $t$-privacy guarantee, the two query distributions are statistically indistinguishable, with perfect equivalence being a strict special case of statistical indistinguishability. Since the adversary’s entire view is determined solely by these queries, the hybrids’ output distributions are also statistically indistinguishable, giving:
    \[
        |\Pr[W_1] - \Pr[W_2]| \leq \mathrm{negl}(\kappa) \quad (\text{i.e., } \mathrm{H}_1 \approx_s \mathrm{H}_2),
    \]
    where $\mathrm{negl}(\kappa)$ denotes a negligible function in $\kappa$.

    \item 
    Hybrid \( \mathrm{H}_3 \) is fully equivalent to \( \mathrm{H}_2 \): \( \mathrm{H}_2 \)’s query generation matches \( \mathsf{Sim}_0 \), and its bit \( b \) computation corresponds to \( \mathsf{Sim}_1 \). All operations in \( \mathrm{H}_2 \) map one-to-one to \( \mathrm{H}_3 \), so they implement the same random process:
    \[
        \Pr[W_2] = \Pr[W_3] \quad (\text{i.e., } \mathrm{H}_2 \equiv \mathrm{H}_3).
    \]

\end{itemize}

Combining the indistinguishability of consecutive hybrids, we have:
\[
    |\Pr[W_0] - \Pr[W_3]| \leq \epsilon + \mathrm{negl}(\kappa) = \mathrm{negl}_2(\kappa),
\]
where the inequality follows from the hybrid bounds, and the equality holds as $\epsilon$ is negligible (by the theorem’s premise) and sums of negligible functions are negligible ($\mathrm{negl}_2(\kappa)$ denotes the combined function).

Recall that $\mathsf{REAL}_{\alpha} = \mathrm{H}_0$ and $\mathsf{IDEAL}_{\mathsf{Sim}} = \mathrm{H}_3$. This directly implies that
\[
    \mathsf{REAL}_{\alpha} \approx_s \mathsf{IDEAL}_{\mathsf{Sim}}.
\]
We thus conclude that the scheme $\Pi$ satisfies the statistical privacy requirement of itAPIR (Definition \ref{def:itapir-privacy}), including resistance to selective-failure attacks. This completes the proof of the theorem.

\end{proof}

\begin{theorem} \label{theo:itapir-construction1}
There exists an $\ell$-server itAPIR scheme satisfying the following properties:
    \begin{itemize}
        \item Statistical $t$-privacy against selective-failure attacks,
        \item $(v, \epsilon)$-integrity,
        \item $O\left( \frac{\ell^2}{t} \left( \frac{n\ell}{t} \right)^{1/(\lfloor (2\ell-1)/t \rfloor - 1)} \log p \right)$ communication complexity,
    \end{itemize}
    where $2 \leq \ell \leq 5$, $1 \leq t < \ell$, $v = \ell-1$, $\epsilon = \frac{3}{p-2}$, and $p$ denotes the size of a finite field.
\end{theorem}

\begin{proof}
The desired itAPIR scheme is constructed by leveraging our conversion theorem (Theorem \ref{theo:pir-rv-to-apir}) and the information-theoretic PIR-RV (itPIR-RV) construction in \cite{ke2023private}. 
\end{proof}

\begin{theorem} \label{theo:itapir-construction2}
There exists an $\ell$-server itAPIR scheme satisfying the following properties:
\begin{itemize}
    \item Statistical $t$-privacy against selective-failure attacks,
    \item $(v, \epsilon)$-integrity,
    \item $n^{\lfloor(2(\ell-t)-1)/t\rfloor^{-1}} (\log n + \log \epsilon^{-1}) (\ell-t)^{1+o(1)} \ell^{2+o(1)}$ communication complexity,
\end{itemize}
where $\ell \geq 2 $, $1 \leq t <  \ell/2$, $v = t$, $\epsilon > 0$.
\end{theorem}
\begin{proof}
The desired itAPIR scheme is constructed by leveraging our conversion theorem (Theorem \ref{theo:pir-rv-to-apir}) and the homomorphic MAC-based PIR construction in \cite{eriguchi2022optimal}. 
\end{proof}

% \noindent\textbf{Summary of This Section.}
% We establish a rigorous conversion bridge between itPIR-RV and itAPIR: any perfect/statistical $t$-private itPIR-RV scheme with $(v,\epsilon)$-integrity ($v \geq t$, $\epsilon$ negligible) upgrades to a statistically $t$-private itAPIR scheme, and we present two concrete constructions via existing itPIR-RV designs \cite{ke2023private, eriguchi2022optimal} to verify the theorem’s practicality.

{
% \color{blue}
\section{Conclusion}
We introduced information-theoretic Authenticated PIR (itAPIR), formalized its relation to itPIR-RV, and showed that any suitable itPIR-RV scheme can be converted into a secure itAPIR scheme. Based on this result, we obtained two concrete itAPIR constructions for small-server and general-server settings, both achieving sublinear communication complexity.

Overall, our work clarifies the landscape of information-theoretic malicious-resilient PIR and provides an assumption-free approach to constructing itAPIR from existing itPIR-RV schemes \cite{ke2023private, eriguchi2022optimal}, with potential relevance to quantum-resistant PIR in malicious environments.
}

\newpage

\newpage

\bibliographystyle{IEEEtran}
\bibliography{ref}

\vspace{12pt}

\end{document}